\documentclass[12pt]{article}

\usepackage[utf8]{inputenc}
\usepackage[T1]{fontenc}

\usepackage{geometry}

\usepackage{amsthm, amsmath, amssymb}

\usepackage{float}
\usepackage{xcolor}
\usepackage{thmtools}
\usepackage{thm-restate}
\usepackage{enumerate}
\usepackage{hyperref}
\usepackage{graphicx}
\usepackage{complexity}

\usepackage{tabularx}
\usepackage{environ}
\makeatletter
\newcommand{\problemtitle}[1]{\gdef\@problemtitle{#1}}%
\newcommand{\problemparameter}[1]{\gdef\@problemparameter{#1}}
\newcommand{\probleminput}[1]{\gdef\@probleminput{#1}}
\newcommand{\problemquestion}[1]{\gdef\@problemquestion{#1}}
\NewEnviron{problemgen}{
\problemtitle{}\problemparameter{}\probleminput{}\problemquestion{}
\BODY
\par\addvspace{.5\baselineskip}
\noindent
\begin{tabularx}{\textwidth}{@{\hspace{\parindent}} l X c}
    \multicolumn{2}{@{\hspace{\parindent}}l}{\@problemtitle} \\
    \textbf{Input:} & \@probleminput \\
    \textbf{Output:} & \@problemquestion
  \end{tabularx}
  \par\addvspace{.5\baselineskip}
}

\usepackage[textsize=small, color=red!30]{todonotes}

\hypersetup{
  pdftitle = {Minimal dominating sets enumeration with FPT-delay parameterized by the degeneracy and maximum degree},
  pdfauthor = {Valentin Bartier, Oscar Defrain, and Fionn Mc Inerney},
  colorlinks = true,
  linkcolor = black!30!red,
  citecolor = black!30!green
}

\newtheorem{theorem}{Theorem}[section]

\newtheorem{lemma}[theorem]{Lemma}
\newtheorem*{lemma*}{Lemma}
\newtheorem{observation}[theorem]{Observation}

\newtheorem{question}[theorem]{Question}

\newtheorem{definition}{Definition}

\newtheorem{remark}[theorem]{Remark}

\newcommand{\transenum}{\textsc{Trans-Enum}} 
\newcommand{\domenum}{\textsc{Dom-Enum}} 
\newcommand{\algoa}{{\sf A}} 
\newcommand{\algob}{{\sf B}} 
\newcommand{\N}{\mathbb{N}}
\renewcommand{\NG}{\mathcal{N}(G)}

\renewcommand{\H}{\mathcal{H}}

\DeclareMathOperator{\inc}{{\sf inc}}
\DeclareMathOperator{\parent}{{\sf parent}}
\DeclareMathOperator{\children}{{\sf children}}
\DeclareMathOperator{\priv}{{\sf priv}}

\newcommand{\intv}[2]{\left \{ #1, \dots, #2 \right \}}

\def\figureHypergraphHeight{4.2cm}
\def\figureTreeHeight{4.5cm}
\def\lncsqed{}

\renewenvironment{abstract}
{\small\vspace{-1em}
\begin{center}
\bfseries\abstractname\vspace{-.5em}\vspace{0pt}
\end{center}
\list{}{
\setlength{\leftmargin}{0.4in}%
\setlength{\rightmargin}{\leftmargin}}%
\item\relax}
{\endlist}

\title{%
Hypergraph dualization with \FPT-delay parameterized by the degeneracy and dimension%
\thanks{This work was supported by the ANR project DISTANCIA (ANR-17-CE40-0015) and the Austrian Science
Fund (FWF, project Y1329).}}
\author{
Valentin Bartier\thanks{Univ Lyon, CNRS, INSA Lyon, UCBL, Centrale Lyon, Univ Lyon 2, LIRIS, UMR5205, F-69622 Villeurbanne, France}\and
Oscar Defrain\thanks{Aix-Marseille Université, CNRS, LIS, 13009 Marseille, France}\and 
Fionn Mc Inerney\thanks{Algorithms and Complexity Group, TU Wien, Austria}}

\date{April 22, 2024} 

\begin{document}

\def\longVersion{}

\maketitle

\begin{abstract}
At STOC 2002, Eiter, Gottlob, and Makino presented a technique called \emph{ordered generation} that yields an $n^{O(d)}$-delay algorithm listing all minimal transversals of an $n$-vertex hyqpergraph of degeneracy~$d$. 
Recently at IWOCA 2019, Conte, Kanté, Marino, and Uno asked whether this \XP-delay algorithm parameterized by $d$ could be made \FPT-delay for a weaker notion of degeneracy, or even parameterized by the maximum degree $\Delta$, i.e., whether it can be turned into an algorithm with delay $f(\Delta)\cdot n^{O(1)}$ for some computable function~$f$. 
Moreover, and as a first step toward answering that question, they note that they could not achieve these time bounds even for the particular case of minimal dominating sets enumeration.
In this paper, using ordered generation, we show that an \FPT-delay algorithm can be devised for minimal transversals enumeration parameterized by the degeneracy and dimension, giving a positive and more general answer to the latter question.

\vskip5pt\noindent{}{\bf Keywords:} algorithmic enumeration,  hypergraph dualization, minimal transversals, minimal dominating sets, \FPT-delay enumeration.
\end{abstract}

\section{Introduction}

Graphs and hypergraphs are ubiquitous in
discrete mathematics and theoretical computer science. 
Formally, a \emph{hypergraph} $\H$ consists of a family $E(\H)$ of subsets called \emph{edges} (or \emph{hyperedges}) over a set $V(\H)$ of elements called \emph{vertices}, and it is referred to as a \emph{graph} when each of its edges is of size precisely two.
A \emph{transversal} in a hypergraph is a set of vertices that intersects every edge.
The problem of enumerating the family $Tr(\H)$ of all (inclusion-wise) minimal transversals of a given hypergraph~$\H$, usually denoted by \transenum{} and also known as \emph{hypergraph dualization}, is encountered in many areas of computer science including logic, database theory, and artificial intelligence~\cite{kavvadias1993horn,eiter1995identifying,gunopulos1997data,eiter2008computational,nourine2012extending}.
The best-known algorithm to date for the problem is due to Fredman and Khachiyan~\cite{fredman1996complexity}, and runs in time $N^{o(\log N)}$, where $N=|V(\H)|+|E(\H)|+|Tr(\H)|$.
The question of the existence of an output-polynomial-time algorithm has been open for over 30 years,
and is arguably one of the most important open problems in algorithmic enumeration \cite{johnson1988generating,eiter2008computational}.
An enumeration algorithm is said to run in \emph{output-polynomial} time if it outputs all solutions and stops in a time that is bounded by a polynomial in the combined sizes of the input and the output.
It is said to run with {\em polynomial delay} if the run times before the first output, after the last output,\footnote{Some authors equivalently choose to ignore the postprocessing by assuming that the last solution is kept in memory and output just before halting, see, e.g., \cite{coussinet}.} and between two consecutive outputs are bounded by a polynomial depending on the size of the input only.
Clearly, any polynomial-delay algorithm defines an output-polynomial-time algorithm.
We refer the reader to \cite{johnson1988generating,strozecki2019survey} for more details on the usual complexity measures of enumeration algorithms, and to~\cite{eiter2008computational} for a survey on minimal transversals enumeration.

In 2002, an important step toward the characterization of tractable cases of \transenum{} was made by Eiter, Gottlob, and Makino \cite{eiter2002new}.
The problem, studied under its equivalent formulation of monotone Boolean dualization, was shown to admit a polynomial-delay algorithm whenever the hypergraph has bounded degeneracy, for an appropriate notion of hypergraph degeneracy (see Section~\ref{sec:prelims} for definitions of graph and hypergraph degeneracies) that generalizes graph degeneracy.
Specifically, the authors proved the following result. 

\begin{theorem}[\cite{eiter2002new,eiter2003new}]\label{thm:eiter}
   There is an algorithm that, given an $n$-vertex hypergraph of weak degeneracy $d$, generates its minimal transversals with $n^{O(d)}$-delay.
\end{theorem}

In particular, the above theorem captures $\alpha$-acyclic hypergraphs and several other classes of hypergraphs related to restricted cases of conjunctive normal form (CNF) formulas.
Another consequence of Theorem~\ref{thm:eiter} is a polynomial-delay algorithm for the intimately related problem of enumerating minimal dominating sets in graphs, usually denoted by \domenum{}, when the graph has bounded degeneracy. 
A \emph{minimal dominating set} in a graph is a (inclusion-wise) minimal subset of vertices intersecting each closed neighborhood.
As the family of closed neighborhoods of a graph defines a hypergraph, \domenum{} naturally reduces to \transenum{}.
However, quite surprisingly, and despite the fact that not all hypergraphs are hypergraphs of closed neighborhoods of a graph, the two problems were actually shown to be equivalent in \cite{kante2014split}.
This led to a significant interest in \domenum{}, and the characterization of its complexity status in various graph classes has since been explored in the literature~\cite{kante2015chordal,golovach2016chordalbip,golovach2018lmimwidth,bonamy2019kt}.

The theory of parameterized complexity provides a framework to analyze the run time of an algorithm as a function of the input size and some parameters of the input such as degeneracy.
It has proved to be very useful in finding tractable instances of \NP-hard problems, as well as for devising practical algorithms. 
In this context, a problem is considered tractable for a fixed parameter $k$ if it is in \emph{\FPT\ when parameterized by $k$}, that is, if it admits an algorithm running in time $f(k)\cdot n^{O(1)}$ for some computable function $f$, where $n$ is the size of the input. 
It is important to note that the combinatorial explosion of the problem can be restrained to the parameter $k$ here.
There are also weaker notions of tractability in this context, with a problem being in \emph{\XP\ when parameterized by $k$} if it admits an algorithm running in time $n^{f(k)}$.
Studying the tractability of important problems within the paradigm of parameterized complexity has led to a fruitful line of research these last decades, and we refer the reader to~\cite{cygan2015parameterized} for an overview of techniques and results in this field.

The study of algorithmic enumeration from a parameterized complexity perspective is quite novel. 
It may be traced back to works such as \cite{fernau2002parameterized,damaschke2006parameterized}, where the authors devise algorithms listing all solutions of size $k$ in time $f(k)\cdot n^{O(1)}$, where $n$ is the size of the input.
Note, however, that these kinds of results hint at a relatively low number of solutions (typically not superpolynomial in $n$ for fixed $k$), while this is not the case for many enumeration problems including the one that we consider in this paper.
New parameterized paradigms for enumeration problems of this second type were later developed in \cite{creignou2017paradigms,golovach2022refined}, and a growing interest for \FPT-enumeration has since been developing as witnessed by the recent Dagstuhl Seminar on output-sensitive, input-sensitive, parameterized, and approximative enumeration~\cite{fernau2019algorithmic}, and the recent progress on this topic~\cite{meier2018enumeration,conte2019maximal,creignou2019parameterised,meier2020incremental,golovach2022refined}.
Among the different complexity measures one could want to satisfy, the notion of \FPT-delay naturally arises.
An algorithm \algoa{} for an enumeration problem parameterized by $k$ runs with \emph{\FPT-delay} if there exists a computable function $f$ such that \algoa{} has delay $f(k)\cdot n^{O(1)}$, where $n$ in the size of the input. An \XP-delay algorithm for an enumeration problem parameterized by $k$ is analogously defined with the delay being $n^{f(k)}$ instead.
Note that the algorithm given by Theorem~\ref{thm:eiter} is \XP-delay, but not \FPT-delay.
The next question naturally arises within the context of parameterized complexity.

\begin{question}\label{qu:big-question}
    Can the minimal transversals of an $n$-vertex hypergraph of weak degeneracy $d$ be listed with delay
    $f(d) \cdot n^{O(1)}$ 
    for some computable function~$f$?
\end{question}

To the best of our knowledge, no progress has been made on this question to date.
At IWOCA 2019, Conte, Kanté, Marino, and Uno formulated the following weakening by considering a weaker notion of degeneracy $d$ and the maximum degree~$\Delta$ as an additional parameter.
Specifically, they asked the following.

\begin{question}[{\cite{conte2019maximal}}]\label{qu:conte}
    Can the minimal transversals of an $n$-vertex hypergraph be listed with delay
    $2^d \cdot \Delta^{f(d)} \cdot n^{O(1)}$ 
    for some computable function $f$, 
    where $d$ is the strong degeneracy and $\Delta$ is the maximum degree of the hypergraph?
\end{question}

As $d\leq \Delta$ (see Section~\ref{sec:prelims} for the definition of strong degeneracy) such a time bound is \FPT{} parameterized by the maximum degree only.
This formulation, however, is not arbitrary. 
It is chosen to match the complexity of an \FPT-delay algorithm the authors in~\cite{conte2019maximal} give for another related problem---namely the enumeration of maximal irredundant sets in hypergraphs---as well as giving a finer analysis of the dependence on $d$ and $\Delta$.
A set of vertices is \emph{irredundant} if removing any of its vertices decreases the number of edges it intersects.
It is easy to see that minimal transversals are maximal irredundant sets. 
However, the converse is not true, and the two problems of listing minimal transversals and maximal irredundant sets behave differently. 
Indeed, while \transenum{} is known to admit an output quasi-polynomial-time algorithm~\cite{fredman1996complexity}, it was announced in~\cite{boros2016wepa} that the enumeration of maximal irredundant sets is impossible in output-polynomial time unless $\P=\NP$. 
As put in \cite{conte2019maximal}, answering Question~\ref{qu:conte} would help to precise the links between the two problems, as maximal irredundancy is commonly believed to be more difficult than minimal transversality.
In the same paper, the authors state that their algorithm for maximal irredundancy does not apply to minimal transversality, and that they were also unable to turn it into one for graph domination with the same delay.
The next question naturally follows, and may be regarded as a weakening of Question~\ref{qu:conte} that is implicit in \cite{conte2019maximal}.

\begin{question}[{\cite{conte2019maximal}}]\label{qu:conte-dom}
    Can the minimal dominating sets of an $n$-vertex graph be enumerated with delay 
    $2^d \cdot \Delta^{f(d)} \cdot n^{O(1)}$ 
    for some computable function $f$, 
    where $d$ is the degeneracy and $\Delta$ is the maximum degree of the graph?
\end{question}

In this paper, we give a positive answer to Question~\ref{qu:conte-dom}, suggesting that minimal domination and maximal irredundancy are of relatively comparable tractability as far as the degeneracy is concerned.
In fact, we prove a stronger result. That is, we design an \FPT-delay algorithm for \transenum{} parameterized by the degeneracy and the maximum size of a hyperedge, usually referred to as \emph{dimension}, even for the notion of hypergraph degeneracy considered in~\cite{eiter2002new,eiter2003new}, which is less restrictive than the one considered in \cite{conte2019maximal}, a point that is discussed in Section~\ref{sec:prelims}.
More formally, we prove the following result.

\begin{theorem}\label{thm:main}
    The minimal transversals of an $n$-vertex hypergraph of weak degeneracy $d$ and dimension $k$ can be enumerated with delay $k^d \cdot n^{O(1)}$.
\end{theorem}

The \FPT-delay algorithm we obtain is based on the ordered generation technique originally presented by Eiter, Gottlob, and Makino in~\cite{eiter2002new}, combined with elementary arguments on the number of minimal subsets hitting the neighborhood of a well-chosen vertex of bounded ``backdegree''.
As a corollary of Theorem~\ref{thm:main}, we obtain an algorithm solving \domenum{} running within the time bounds required in Question~\ref{qu:conte-dom}.
We then discuss the limitations of the technique and show that it cannot be directly used to get a positive answer to Question~\ref{qu:big-question} unless $\FPT=\W[1]$ (see Section~\ref{sec:discussion} for more details and references). 

Additionally, we believe that our presentation is of valuable interest as far as the ordered generation technique and its limitations in the design of parameterized algorithms are concerned.
Furthermore, we hope that it will raise interest in making headway in the recent and promising line of research concerning the study of enumeration problems in the parameterized setting.

\paragraph{Organization of the paper.} In the next section, we introduce the problem and related notation, and discuss the aforementioned variants of hypergraph degeneracy considered in \cite{eiter2003new,conte2019maximal}.
In Section~\ref{sec:og}, we present the ordered generation technique applied to hypergraphs, and show how it amounts to generating extensions of a partial solution while preserving \FPT-delay.
In Section~\ref{sec:children-generation}, we show how to generate extensions with \FPT-delay parameterized by the weak degeneracy and dimension, proving Theorem~\ref{thm:main}.
The answer to Question~\ref{qu:conte-dom} is given in Section~\ref{sec:dom} as a corollary.
Future research directions and the limitations of the presented techniques are discussed in Section~\ref{sec:discussion}.

\section{Preliminaries}\label{sec:prelims}

\paragraph{Graphs and hypergraphs.}
For a hypergraph $\H$, its vertex set is denoted by $V(\H)$, and its edge set is denoted by $E(\H)$. 
If every edge has size precisely two, then it is referred to as a graph and is usually denoted by $G$.
Given a vertex $v\in V(\H)$, the set of edges that contain $v$ (also called edges \emph{incident} to $v$) is denoted by $\inc(v)$.
The \emph{degree} of $v$ in $\H$, denoted by $\deg_{\H}(v)$, is the number of edges in $\inc(v)$.
The maximum degree of a vertex in $\H$ is denoted by $\Delta(\H)$ (or simply $\Delta$).
The \emph{dimension} of $\H$ is the maximum size of an edge in $\H$.
A \emph{transversal} of $\H$ is a subset $T\subseteq V(\H)$ such that $E\cap T\neq \emptyset$ for all $E\in E(\H)$.
It is \emph{minimal} if it is minimal by inclusion, or equivalently, if $T\setminus\{v\}$ is not a transversal for any $v\in T$.
The set of minimal transversals of $\H$ is denoted by $Tr(\H)$.

In this paper, we are interested in the following problem known to admit an \XP-delay algorithm~\cite{eiter2003new} parameterized by the degeneracy, but for which the existence of an \FPT-delay algorithm parameterized by the same parameter (even in addition to the maximum degree) is open~\cite{conte2019maximal}.

\begin{problemgen}
  \problemtitle{\textsc{Minimal Transversals Enumeration} (\transenum{})}
  \probleminput{A hypergraph $\H$.}
  \problemquestion{The set $Tr(\H)$.}
\end{problemgen}

We introduce notation that is convenient for minimal transversality.
Let $S\subseteq V(\H)$ and let $v$ be a vertex in $S$.
A \emph{private edge of $v$ with respect to $S$} is an edge $E\in E(\H)$ such that $E\cap S=\{v\}$. 
The set of private edges of $v$ with respect to $S$ are denoted by $\priv(S,v)$.
It is well-known that a transversal in a hypergraph is minimal if and only if each vertex it contains has a private edge.


The \emph{degeneracy} of a graph is a well-known parameter that has received considerable attention as a measure of sparsity in a graph \cite{nevsetvril2012sparsity}, as well as a suitable parameter to devise tractable graph algorithms \cite{eiter2003new,alon2009linear,philip2012polynomial,wasa2014efficient}.
For any graph $G$, it corresponds to the least integer $d$ such that every subgraph of $G$ has a vertex of degree at most $d$, and it can be computed in time $O(|V(G)|+|E(G)|)$~\cite{MatulaBeck83}.
Several generalizations to hypergraphs have been proposed in the literature \cite{eiter2003new,kostochka2008adapted,bar2010online,dutta2016subgraphs}.
Among these generalizations, two of them naturally arise depending on the notion of subhypergraph we consider, and were previously considered in the context of hypergraph dualization~\cite{eiter2003new,conte2019maximal}.
We review them here.

Let $\H$ be a hypergraph and $S$ a subset of vertices of $\H$. 
The hypergraph \emph{induced} by $S$ is the hypergraph 
\[
    \H[S]:=(S, \{E\in E(\H) : E\subseteq S\}).
\]
The \emph{trace} of $S$ on $\H$ is the hypergraph 
\[
    \H_{|S}:=(S, \{E \cap S : E\in E(\H),\ E\cap S\neq \emptyset\}).
\]
Several edges in $\H$ may have the same intersection with $S$, but this is only counted once in $\H_{|S}$, i.e., the trace is a set and not a multiset. 
An example of these notions of subhypergraph is in Figure~\ref{fig:hypergraph}.
The two following definitions generalize graph degeneracy (with the strong degeneracy being at most one more than the graph degeneracy), and differ on whether we consider the subhypergraph to be the hypergraph induced by the first $i$ vertices or its trace on these vertices.

\begin{figure}
    \centering
    \includegraphics[height=\figureHypergraphHeight{}]{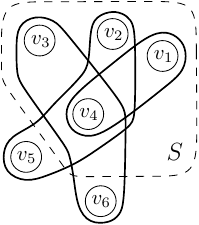}
    \caption{A hypergraph $\H$ on $6$ vertices and $3$ edges. The hypergraph induced by $S:=\{v_1,v_2,v_3,v_4\}$ has a single edge $\{v_1,v_4\}$, while the trace of $S$ on $\H$ contains three edges $\{v_1,v_4\}$, $\{v_2,v_4\}$, and $\{v_3,v_4\}$.}
    \label{fig:hypergraph}
\end{figure}

\begin{definition}[Weak degeneracy, \cite{eiter2003new}]\label{def:strong-degeneracy}
    The weak degeneracy of $\H$ is the smallest integer $d$ such that there exists an ordering $v_1,\dots, v_n$ of its vertices satisfying for every $1\leq i \leq n$: $$\deg_{\H[\{v_1,\dots, v_i\}]}(v_i) \leq d.$$ 
\end{definition}

\begin{definition}[Strong degeneracy, \cite{kostochka2008adapted,conte2019maximal}]\label{def:weak-degeneracy}
    The strong degeneracy of $\H$ is the smallest integer $d$ such that there exists an ordering $v_1,\dots, v_n$ of its vertices satisfying for every $1\leq i \leq n$: $$\deg_{\H_{|\{v_1,\dots, v_i\}}}(v_i) \leq d.$$ 
\end{definition}

The next inequality follows from the fact that, for any subset $S$ of $V(\H)$, an edge $E$ that belongs to $\H[S]$ also belongs to $\H_{|S}$ by definition.

\begin{remark}\label{rem:degeneracy-parameters}
    A hypergraph's weak degeneracy is at most its strong degeneracy.
\end{remark}

Note that the strong degeneracy of a hypergraph can be arbitrarily larger than its weak degeneracy. To see this, consider the hypergraph $\H$ formed by taking a clique on $k$ vertices $v_1,\dots,v_k$ (a graph), and adding one new vertex $u$ in each edge, so that each hyperedge in $\H$ has size $3$. 
Let $U$ be the set of vertices $u$ added in the construction, and consider any ordering of $V(\H)$ such that the elements $v_1,\dots,v_k$ appear before those in $U$.
The weak degeneracy of $\H$ is $1$ since each vertex $u$ has degree $1$ and the hypergraph induced by $v_1,\dots,v_k$ has no edge. On the other hand, the strong degeneracy of $\H$ is at least $k$ since the trace of $\H$ on $v_1,\dots,v_k$ is a clique whose vertices all have degree $k$.

Also, note that while there exist $n$-vertex hypergraphs with $O(2^n)$ edges, any hypergraph that has (weak or strong) degeneracy $d$ contains $O(dn)$ edges. 
This explains why time 
bounds such as those expected in Questions~\ref{qu:conte}~and~\ref{qu:conte-dom}, and obtained in Theorem~\ref{thm:main} and subsequent lemmas in this paper are reasonable.

For both notions of degeneracy, an \emph{elimination ordering} is an order $v_1,\dots, v_n$ as in their definitions.
Recall that Question~\ref{qu:conte} is stated for strong degeneracy~\cite{conte2019maximal}.
In this paper, we show Theorem~\ref{thm:main} to hold for weak degeneracy, which implies the same result for strong degeneracy.
Note that the weak (strong, respectively) degeneracy of a hypergraph $\H$ along with an elimination ordering witnessing that degeneracy can be computed in time $O(|V(\H)|+\sum_{E\in E(\H)}|E|)$~\cite{MatulaBeck83,Shahrokhi22}.

\paragraph{Graph domination.}
If a hypergraph is a graph, then its transversals are usually referred to as vertex covers in the literature.
This is not to be confused with domination, that we define now.
Let $G$ be a graph.
The \emph{closed neighborhood} of a vertex $v$ of $G$ is the set $N[v]:=\{v\}\cup \{u : \{u,v\} \in E(G)\}$.
A \emph{dominating set} in $G$ is a subset $D\subseteq V(G)$ intersecting every closed neighborhood in the graph.
It is called \emph{minimal} if it is minimal by inclusion, and the set of minimal dominating sets of $G$ is denoted by $D(G)$.
In the minimal dominating sets enumeration problem, denoted \domenum{}, the task is to generate $D(G)$ given $G$.

As stated in the introduction, \transenum{} and \domenum{} are polynomially equivalent, one direction being a direct consequence of the easy observation that $D(G)=Tr(\mathcal{N}(G))$ for any graph $G$, where $\mathcal{N}(G)$ denotes the \emph{hypergraph of closed neighborhoods} of $G$ whose vertex set is $V(G)$ and whose edge set is $\{N[v] : v\in V(G)\}$ (see, e.g.,~\cite{berge1984hypergraphs}).
The other direction was exhibited by Kanté, Limouzy, Mary, and Nourine in~\cite{kante2014split} even when restricted to co-bipartite graphs.
However, the construction in \cite{kante2014split} does not preserve the degeneracy.
Hence, the existence of an \FPT-delay algorithm for \domenum{} parameterized by the degeneracy (and more) would not directly imply one for \transenum. 

\section{Ordered generation}\label{sec:og}

Our algorithm is based on the \emph{ordered generation} technique introduced in \cite{eiter2002new} for monotone Boolean dualization, and later adapted to the context of graph domination in \cite{bonamy2019triangle,bonamy2019kt}.
Our presentation of the algorithm is based on these works.

Let $\H$ be a hypergraph and $v_1, \dots, v_n$ an ordering of its vertex set.
For readability, given any $i\in \intv{1}{n}$, we set $V_i:=\{v_1,\dots,v_i\}$, 
    $\H_i:=\H[V_i]$,
and denote by $\inc_i(u)$, for any $u\in V(\H_i)$, the incident edges of $u$ in $\H_i$.
By convention, we set $V_0:=\emptyset$, $\H_0:=(\emptyset, \emptyset)$, and $Tr(\H_0):=\{\emptyset\}$.
Further, for all $0\leq i\leq n$, $S\subseteq V(\H)$, and $v\in S$, we denote by $\priv_i(S,v)$ the private edges of $v$ with respect to $S$ in $\H_i$.
Let us now consider some integer $i\in \intv{0}{n-1}$. 
For $T\in Tr(\H_{i+1})$, we call the \emph{parent of $T$ with respect to $i+1$}, denoted by $\parent(T, i+1)$, the set obtained by applying the following greedy procedure: 

\begin{quote}
    \emph{While there exists a vertex $v\in T$ such that $\priv_i(T, v) = \emptyset$,\\ remove from $T$ a vertex of smallest index with that property.}
\end{quote}

Observe that $\parent(T, i+1)$ is uniquely defined by this routine. 
Given a minimal transversal $T^*$ of $\H_{i}$, we call the \emph{children of $T^*$ with respect to $i$}, denoted by $\children(T^*, i)$, the family of sets $T\in Tr(\H_{i+1})$ such that $\parent(T, i+1)=T^*$. 

\begin{lemma}\label{lem:parent}
    Let $0\leq i\leq n-1$,
    $T\in Tr(\H_{i+1})$. 
    Then, $\parent(T,i+1)\in Tr(\H_i)$.
\end{lemma}

\begin{proof}
    By definition, every edge of $\H_i$ is an edge of $\H_{i+1}$.
    Hence, $T$ is a transversal of $\H_i$ as well.
    Now, as the greedy procedure above only removes vertices $v$ that have no private edge in $\H_i$, the obtained set is a transversal of $\H_i$ at each step, and ends up minimal by construction.
\lncsqed{}\end{proof}

\begin{lemma}\label{lem:child}
    Let $0\leq i\leq n-1$ and $T^*\in Tr(\H_i)$.
    Then, either:
    \begin{itemize}
        \item $T^*\in Tr(\H_{i+1})$~and $\parent(T^*, i+1) = T^*$; or
        \item $T^*\cup \{v_{i+1}\}\in Tr(\H_{i+1})$ and~$\parent(T^*\cup \{v_{i+1}\}, i+1) = T^*$.
    \end{itemize}
\end{lemma}

\begin{proof}
    As $T^*\in Tr(\H_i)$, we have $\priv_i(T^*,v)\neq\emptyset$ for every $v\in T^*$, and since every edge of $\H_i$ is an edge of $\H_{i+1}$, we derive that $\priv_{i+1}(T^*,v)\neq\emptyset$ for every $v\in T^*$.
    Hence, if $T^*$ is a transversal of $\H_{i+1}$, then we obtain $T^*\in Tr(\H_{i+1})$, and thus, $\parent(T^*, i+1) = T^*$ by definition, yielding the first item of the lemma.

    Let us thus assume that $T^*$ is not a transversal of $\H_{i+1}$.
    Then, there is an edge $E$ in $\H_{i+1}$ that is not hit by $T^*$.
    As $T^*$ is a transversal of $\H_i$, that edge must contain the vertex $v_{i+1}$.
    We derive that $T:=T^*\cup \{v_{i+1}\}$ is a transversal of $\H_{i+1}$, which furthermore satisfies $\priv_{i+1}(T, v_{i+1})\neq\emptyset$ as $E$ is not hit by $T^*$. 
    Let us show that $\priv_{i+1}(T, v)\neq\emptyset$ holds for any other vertex $v\in T$. 
    Consider one such vertex $v\neq v_{i+1}$ and $E\in \priv_i(T^*, v)$. 
    Such an edge exists as $T^*\in Tr(\H_i)$.
    If $E\not\in \priv_{i+1}(T, v)$, then we have that $v_{i+1}\in E$, which contradicts the fact that $E$ belongs to $\H_i$.
    Hence, $E\in \priv_{i+1}(T, v)$, and the lemma follows.
\lncsqed{}\end{proof}

The $\parent$ relation defines a rooted tree, called the \emph{solution tree}, with vertex set $\{(T, i): T\in Tr(\H_i),\ 0\leq i\leq n\}$, leaves $\{(T,n) \mid T \in Tr(\H)\}$, and root $(\emptyset,0)$.
Our generation algorithm consists in a traversal of that tree, starting from its root, that outputs every leaf.
Lemma~\ref{lem:parent} ensures that each node $(T, i)$, $0<i\leq n$, can be obtained from its parent.
Lemma~\ref{lem:child} ensures that each node $(T, i)$, $0\leq i<n$, has at least one child, i.e., that every branch of the tree leads to a different minimal transversal of $\H$.
In what follows, the set $T$ of an internal node $(T,i)$ is referred to as a \emph{partial solution}; see Figure~\ref{fig:tree} for a representation of the tree.
The complexity of the algorithm depends on the delay 
needed to generate the children of an internal node, which is formalized in the next theorem.

\begin{figure}
    \centering
    \includegraphics[height=\figureTreeHeight{}]{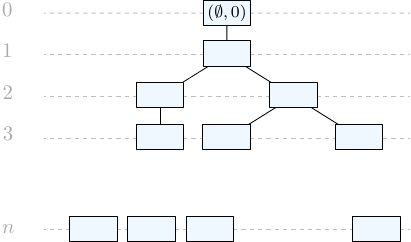}
    \caption{A representation of the solution tree defined by the $\parent$~relation.}
    \label{fig:tree}
\end{figure}

\begin{theorem}\label{the:og}
    Let $f,s\colon \N^2\to \mathbb{Z}^+$ be two functions.
    Suppose there is an algorithm that, given a hypergraph $\H$ on vertices $v_1,\dots,v_n$ and $m$ edges, an integer $i\in \intv{0}{n-1}$, and $T^*\in Tr(\H_i)$, enumerates $\children(T^*,i)$ with delay $f(n,m)$ and using $s(n,m)$ space.
    Then, there is an algorithm that
    enumerates the set $Tr(\H)$ with delay $O(n\cdot f(n,m))$ and total space $O(n \cdot s(n,m))$.
\end{theorem}

\begin{proof}
    We assume the existence of an algorithm \algoa{} of the first type, and describe an algorithm \algob{} listing $Tr(\H)$ with the desired delay. 
    The algorithm first chooses an ordering $v_1,\dots, v_n$ of the vertex set of $\H$ in $O(n)$ time.
    Then, it runs a DFS of the solution tree defined by the parent relation above, starting from its root $(\emptyset, 0)$ and only outputting sets $T$ corresponding to visited leaves $(T, n)$.
    Recall that these sets are exactly the minimal transversals of $\H$.
    First, the correctness of \algob{} follows from Lemmas~\ref{lem:parent} and~\ref{lem:child}.
    Let us now consider the time 
    and space 
    complexity of such an algorithm.
    At first, the root $(\emptyset, 0)$ is constructed in $O(1)$ time
    and space.
    Then, for each internal node $(T^*, i)$ of the tree, the set $\children(T^*, i)$ is generated with delay $f(n, m)$ 
    and space $s(n, m)$ 
    by \algoa{}.
    We furthermore note that when visiting a node of the solution tree, we do not need to generate all its children.
    We can instead compute one child that has not been processed yet, and continue the DFS on that child. Thus, when we visit some node $(T, i)$, we only need to store the data of the $i$ paused executions of \algoa{} enumerating the children of the ancestors of $(T, i)$ including the root $(\emptyset, 0)$, together with the data of the current instance of \algoa{} listing $\children(T, i)$. 
    Hence, the maximum delay of \algob{} between two consecutive outputs is bounded by twice the height of the solution tree times the delay of \algoa{}, and the same goes for the space complexity.
    Since the height of the tree is $n$, we get the desired bounds of $O(n\cdot f(n,m))$ 
    and $O(n\cdot s(n,m))$ 
    for the delay 
    and space complexities 
    of~\algob{}, respectively.
\lncsqed{}\end{proof}


Note that only a multiplicative factor of $f(n, m)$ is added to the input size in the delay of the algorithm given by Theorem~\ref{the:og}. 
Thus, an algorithm generating children with \FPT-delay $f(k)\cdot \poly(n, m)$, for some parameter $k$ and computable function $f$, yields an \FPT-delay algorithm for \transenum{} parameterized by~$k$.

\section{Parameterized children generation}\label{sec:children-generation}

In hypergraphs, we show that the children of an internal node of the solution tree, as defined in Section~\ref{sec:og}, may be enumerated with \FPT-delay parameterized by the weak degeneracy and the dimension.
By Theorem~\ref{the:og}, this leads to an algorithm listing minimal transversals in hypergraphs with \FPT-delay parameterized by the weak degeneracy and the dimension, i.e., Theorem~\ref{thm:main}.

Let $\H$ be a hypergraph of weak degeneracy $d$ and dimension $k$, and let $v_1, \dots, v_n$ be an elimination ordering of $\H$ witnessing the degeneracy. By definition, for any $i\in\intv{0}{n-1}$ and $E\in \inc_{i+1}(v_{i+1})$, we have $|\inc_{i+1}(v_{i+1})|\leq d$ and ${|E|\leq k}$.
We fix $i\in \intv{0}{n-1}$ and $T^*\in Tr(\H_i)$ in what follows.




\begin{lemma}\label{lem:k-private}
    For any $X\subseteq V_{i+1}\setminus T^*$, if $T^*\cup X\in Tr(\H_{i+1})$, then every $x\in X$ has a distinct private edge in $\inc_{i+1}(v_{i+1})$, and thus, $|X|\leq d$.
\end{lemma}

\begin{proof}
    This holds since $|\inc_{i+1}(v_{i+1})|\leq d$ and no edge in $\H_i$ may be the private edge of $x\in X$ with respect to $T^* \cup X$ as all such edges are hit by~$T^*$.
\lncsqed{}\end{proof}


\begin{lemma}\label{lem:children-generation}
    The set $\children(T^*,i)$ can be generated in $k^d \cdot n^{O(1)}$ time.
\end{lemma}

\begin{proof}
    Recall from Lemma~\ref{lem:child} that either $\children(T^*,i)=T^*$ or $T^*\cup \{v_{i+1}\}$ belongs to $\children(T^*,i)$.
    Clearly, these two situations can be sorted out in $n^{O(1)}$ time.
    In the first case, we have generated the children within the claimed time.
    Let us thus assume that $\children(T^*,i)\neq T^*$ and $T^*\cup \{v_{i+1}\}$ has already been produced as a first child in $n^{O(1)}$ time.
    
    We now focus on the generation of sets $X\subseteq V_{i+1}\setminus T^*$ with $v_{i+1}\notin X$ such that $T^*\cup X\in Tr(\H_{i+1})$.
    For each edge $E$ in $\inc_{i+1}(v_{i+1})$, we guess at most one element $x_E\neq v_{i+1}$ in $E$ to be the element of $X$ that has $E$ as its private edge.
	As $\inc_{i+1}(v_{i+1})$ has at most $d$ edges and each of these edges has size at most $k$ and contains $v_{i+1}$, we get at most $(k+1-1)^d$ choices for $X$.
    For every such set $X$, we check whether the obtained set $T:=T^*\cup X$ is a minimal transversal of $\H_{i+1}$ satisfying $T^*=\parent(T,i+1)$ in $n^{O(1)}$ following the definition.
    This takes $k^d \cdot n^{O(1)}$ time in total.
    By Lemma~\ref{lem:k-private}, this procedure defines an exhaustive search for the children of the node $(T^*,i)$ in the solution tree.
\lncsqed{}\end{proof}

Theorem~\ref{thm:main} then follows by Theorem~\ref{the:og} and Lemma~\ref{lem:children-generation}.
In Lemma~\ref{lem:children-generation}, we do not need to store all such sets $X$ as they can be generated following a lexicographic order on the vertices. So, this approach only needs space that is polynomial in~$n$.

\section{Consequences for minimal domination}\label{sec:dom}

Let $G$ be a graph of degeneracy $d$ and maximum degree $\Delta$. 
We show that, as a corollary of Theorem~\ref{thm:main}, we obtain an \FPT-delay algorithm listing the minimal dominating sets of $G$ within the time bounds of Question~\ref{qu:conte-dom}.

In what follows, let $v_1, ..., v_n$ be an elimination ordering of $G$ witnessing the degeneracy, and recall that it can be computed in time $O(|V(G)|+|E(G)|)$~\cite{MatulaBeck83}. Further, recall that $\NG$ denotes the hypergraph of closed neighborhoods of $G$ (see Section~\ref{sec:prelims}).
Note that $v_1, ..., v_n$ witnesses that the weak degeneracy of $\NG$ is at most $d+1$.
Indeed, for any $i\in \intv{1}{n}$, the number of hyperedges in $\inc_i(v_i)$ is bounded by the number of neighbors of $v_i$ in $\{v_1,\dots, v_{i-1}\}$ plus possibly one if the neighborhood of $v_i$ is contained in $\{v_1,\dots, v_{i-1}\}$.
We furthermore have the following relation between the dimension of $\NG$ and $\Delta$.

\begin{observation}\label{obs:degree-dimension}
    The dimension of $\NG$ is at most $\Delta(G)+1$.
\end{observation}

\begin{proof}
    Let $E$ be a hyperedge in $\NG$.
    Then, $E=N[u]$ for some $u\in V(G)$, and hence, $|E|$ is at most the degree of $u$ plus one, which is bounded by $\Delta(G)+1$.
\lncsqed{}\end{proof}

As a corollary of Theorem~\ref{thm:main} and Observation~\ref{obs:degree-dimension}, we obtain a $(\Delta+1)^{d+1}\cdot n^{O(1)}$-time algorithm for \domenum{}, and thus, a $\Delta^{f(d)}\cdot n^{O(1)}$-time algorithm for \domenum{} for a computable function $f$. This answers Question~\ref{qu:conte-dom} in the affirmative.

It should be mentioned that using the Input Restricted Problem technique (IRP for short) from~\cite{cohen2008generating}, more direct arguments yield an \FPT{} algorithm parameterized by $\Delta$.
This is of interest if no particular care on the dependence in the degeneracy is required.
In a nutshell, this technique amounts to reducing the enumeration to a restricted instance consisting of a given solution $S$ and a vertex $v$ not in $S$, in the same (but even more restricted) spirit as in Section~\ref{sec:children-generation}.
Specifically, the authors show that if such an enumeration can be solved in polynomial time (in the size of the input only), then the general problem may be solved with polynomial delay.
It can moreover be seen that such a technique preserves \FPT{} delay, a point not considered in~\cite{cohen2008generating}.
Then, using Observation~\ref{obs:degree-dimension} and noting that the degree of $\NG$ is bounded by $\Delta +1$ as well, the IRP can be solved in $2^{\Delta^2}\cdot n^{O(1)}$ time, and using \cite[Theorems 5.8 and 5.9]{cohen2008generating}, we deduce an \FPT{}-delay algorithm parameterized by $\Delta$ for \domenum{}.
Note that more work is to be done in order to get the dependence in the degeneracy, or for such a technique to be applied to hypergraphs of bounded degeneracy and dimension, namely to obtain Theorem~\ref{thm:main}.
In order to get such time bounds, the algorithm of~\cite{cohen2008generating} must be guided by the degeneracy order.
The interested reader will notice that this guided enumeration is in essence what is proposed by the ordered generation technique, earlier introduced by Eiter et al.~in~\cite{eiter2002new} for hypergraph dualization specifically, while the work of \cite{cohen2008generating} aimed at being more general by considering hereditary properties, with both approaches building up on ideas from~\cite{tsukiyama1977new}.
 
\section{Discussion and limitations}\label{sec:discussion}

We exhibited an \FPT-delay algorithm for \transenum{} parameterized by the weak degeneracy and dimension, yielding an \FPT-delay algorithm for \domenum{} parameterized by the (degeneracy and) maximum degree, answering a question in \cite{conte2019maximal}. 
In light of the existence of an \FPT{}-delay algorithm for \domenum{} parameterized by the maximum degree,
it is tempting to ask whether it admits an \FPT{}-delay algorithm for other structural parameters.
However, for well-studied parameters like the clique number in graphs~\cite{bonamy2019kt} or poset dimension in comparability graphs~\cite{bonamy2020comp}, even the existence of an \XP{}-delay algorithm for \domenum{} remains open.
A natural strengthening of Question~\ref{qu:conte-dom} and weakening of Question~\ref{qu:big-question}, would be to get rid of the maximum degree:

\begin{question}\label{qu:fpt-degen}
    Can the minimal dominating sets of an $n$-vertex $d$-degenerate graph be enumerated with delay
    $f(d) \cdot n^{O(1)}$ 
    for some computable function $f$?
\end{question}

However, it is valuable to note that the techniques presented here seem to fail in the context of hypergraphs, as an \FPT-delay algorithm listing children parameterized by the weak degeneracy (for any elimination ordering witnessing the weak degeneracy) yields an \FPT-delay algorithm for \textsc{Multicolored Independent Set} parameterized by the number of colors, a notorious $W[1]$-hard problem~\cite{cygan2015parameterized}.
Recall that, in \textsc{Multicolored Independent Set}, the vertices of a graph $G$ are colored with colors $1,\dots, k$, and we have to find an independent set containing exactly one vertex from each color.
Formally, we prove the following.

\begin{theorem}
    Let $G$ be an instance of \textsc{Multicolored Independent Set} with colors $1,\dots, k$ with $k,|E(G)| \geq 2$. Then, there exists an $n$-vertex hypergraph $\H$ of weak degeneracy $k+1$ such that, for any elimination ordering $v_1,\dots,v_n$ witnessing the weak degeneracy, there exists a minimal transversal $T^*\in Tr(\H_{n-1})$ whose children are all multicolored independent sets of $G$ except for one.
\end{theorem}

\begin{proof}
We construct an instance of children generation as follows.
We start with the vertices of $G$ and add the vertices $w,u_1,\ldots,u_{k+2}$.
For each color $1\leq \ell\leq k$ and integer $1\leq j\leq k+2$, we create a hyperedge $E_{\ell,j}$ consisting of the vertices of color $\ell$ together with $u_j$.
For each edge $xy\in E(G)$, we add a vertex $z_{xy}$ and two hyperedges $A_{xy}:=\{x,z_{xy}\}$ and $B_{xy}:=\{y,z_{xy}\}$.
Let $Z:=\{z_{xy}: xy\in E(G)\}$.
We add a hyperedge containing $w$, $u_1$, and all the vertices of $Z$.
Lastly, for all $2\leq j\leq k+2$, we add a hyperedge $\{w,u_j\}$ and a hyperedge containing $u_j$ and all the vertices of $Z$.
This completes the construction of the hypergraph $\H$.

We first show that $\H$ has weak degeneracy $k+1$.
The lower bound comes from the minimum degree which is $k+1$.
The upper bound is witnessed by the degeneracy ordering that first removes $u_1,w,u_2,\ldots, u_{k+2}$ in that order, then the $z_{xy}$'s, and then all of the remaining (isolated) vertices in an arbitrary order.

Let $v_1,\dots, v_n$ be any elimination ordering of $\H$ witnessing the weak degeneracy.
Since $u_1$ is the only vertex of minimum degree $k+1$ in $\H$, it must be that $v_n=u_1$.
Thus, $\H_n=\H$ and the hyperedges belonging to $\inc_{n}(v_n)$ are $E_{\ell,1}$ for all $1\leq \ell\leq k$, and the one containing $w$, $u_1$, and all the vertices of $Z$.

Let $T^*:=Z\cup \{u_2,\ldots,u_{k+2}\}$.
By construction, $T^*\in Tr(\H_{n-1})$, $A_{xy}$ and $B_{xy}$ are the only private edges of each $z_{xy}\in T^*$, and $\{w,u_j\}$ is the only private edge of each $u_j\in T^*$.
Consider the generation of the children of $T^*$.
For $X\subseteq V_n\setminus T^*$ to be such that $T^*\cup X\in Tr(\H_n)$, $X$ must hit $E_{\ell,1}$ for all $1\leq \ell\leq k$ (the hyperedge containing $w$, $u_1$, and all the vertices of $Z$ is already hit by vertices in $T^*$).
Further, each such hyperedge must be hit exactly once since $u_1$ hits every such hyperedge and any two other vertices contained in a same such hyperedge have the same trace on $\inc_n(v_n)$. 
Moreover, $X$ should not contain the two endpoints of any $xy\in E(G)$, as otherwise the corresponding vertex $z_{xy}\in T^*$ would lose its private edges.
Hence, among all such possible candidate sets $X$, one consists of the singleton $\{v_n\}$, and the others are multicolored independent sets of~$G$.
\lncsqed{}\end{proof}

Another natural strengthening of Theorem~\ref{thm:main} would be to relax the degeneracy, and require to be \FPT-delay in the dimension only. 
We note however that even the existence of an \XP{}-delay algorithm in that context seems open~\cite{khachiyan2007dualization}.

Finally, for the open questions stated above, an intermediate step would be to aim at \emph{\FPT{}-total-time algorithms}, that is, algorithms producing all the solutions in time $f(k)\cdot N^{O(1)}$ for $k$ the parameter, $f$ some computable function, and $N$ the size of the input plus the output. If successful, it would be interesting to know if these algorithms can be made \emph{\FPT{}-incremental}, that is, if they can produce the $i^\text{th}$ solution in $f(k)\cdot (i+n)^{O(1)}$ time, where $n$ is the size of the input. 
In that direction, \FPT{}-incremental time was obtained in \cite{elbassioni2008some} for the maximum degree.

\paragraph{Acknowledgements.} The second author is thankful to Mamadou Kanté for bringing Question \ref{qu:big-question} to his attention. 
We would like to thank the anonymous reviewers for their careful reading and for providing simpler arguments in Section~\ref{sec:children-generation}.

\bibliographystyle{alpha}
\bibliography{main}

\end{document}